\documentclass[11pt,reqno]{amsart}

\usepackage[margin=1.15in]{geometry}
\usepackage{amsmath,amssymb,amsfonts,amsthm,mathtools}
\usepackage{bm}
\usepackage{enumitem}
\usepackage{hyperref}
\usepackage{microtype}
\usepackage{graphicx}
\usepackage{placeins}

\hypersetup{
  colorlinks=true,
  linkcolor=blue,
  citecolor=blue,
  urlcolor=blue,
  pdftitle={Conditional Euclidean--Hamiltonian Reductions for Sign-Problem Toy Models},
  pdfauthor={Tsogtgerel Gantumur},
  pdfsubject={Conditional Euclidean--Hamiltonian reduction for finite sign-problem models},
  pdfkeywords={sign problem, Euclidean Monte Carlo, effective Hamiltonian,
    projected transfer matrix, finite density}
}

\numberwithin{equation}{section}

\newtheorem{theorem}{Theorem}[section]
\newtheorem{proposition}[theorem]{Proposition}

\theoremstyle{definition}

\newcommand{\Hcal}{\mathcal H}
\newcommand{\Kcal}{\mathcal K}
\newcommand{\Fcal}{\mathcal F}
\newcommand{\Vcal}{\mathcal V}

\newcommand{\Mcal}{\mathcal M}

\newcommand{\Tr}{\operatorname{Tr}}

\newcommand{\e}{\mathrm e}
\newcommand{\ii}{\mathrm i}

\newcommand{\ip}[2]{\left\langle #1,#2\right\rangle}

\title[Conditional Euclidean--Hamiltonian reductions]{Conditional Euclidean--Hamiltonian reductions for sign-problem toy models}

\author{Tsogtgerel Gantumur}
\address{McGill University, Montr\'{e}al, Canada}
\address{National University of Mongolia, Ulaanbaatar, Mongolia}
\address{Mongolian Academy of Sciences, Institute of Mathematics and Digital Technology}
\email{gantumur.tsogtgerel@mcgill.ca}

\date{July 17, 2026}

\begin{document}

\begin{abstract}
Euclidean Monte Carlo methods are effective when the path-integral weight is
real and nonnegative, but finite-density fermion systems often produce
sign-changing or complex scalar weights after the fermionic sector is traced
out.  Hamiltonian formulations avoid this complex-weight sampling problem but
face rapid Hilbert-space growth.  This paper studies a conditional
Euclidean--Hamiltonian (CEH) reduction that combines these two descriptions.

The calculation is organized around a Monte Carlo-tractable reference problem
and a residual active sector.
Instead of tracing the active sector into a determinant
or scalar weight, CEH keeps it operator-valued and uses the reference
calculation to determine projected correlation or transfer matrices.  These
matrices define a finite effective Hamiltonian, with the remaining
finite-density dependence introduced after projection.  At finite rank, the
result is an effective model whose accuracy must be tested through basis
enlargement, metric conditioning, stochastic matrix-element errors, and
number-sector diagnostics.

The construction is examined in three finite benchmarks.  A two-channel
oscillator tests conditional basis compression; a positive-measure stochastic
calculation tests correlation-matrix harvesting and GEVP extraction, including
a comparison with a PDMS-style estimator; 
and a four-site Hubbard ring combines a finite-budget determinant-sign stress
test with finite-density continuation in non-target active spaces.  The
benchmarks support the proposed trade from scalar sign reweighting to a
monitored active-space approximation in structured finite models, but they do
not provide an end-to-end stochastic CEH treatment of the Hubbard sign problem
or establish favorable scaling.  Usefulness requires both low-rank
approximability and efficient extraction of the required projected matrix data.
\end{abstract}

\maketitle


\section{Introduction}
\label{sec:introduction}

Euclidean path integrals provide one of the most successful computational
representations of quantum many-body systems and quantum field theories.  When
the Euclidean weight is real and nonnegative, one may write schematically
\begin{equation}
        Z
        =
        \int \e^{-S_E[\phi]}\,D\phi,
        \qquad
        \e^{-S_E[\phi]}\geq0,
        \label{eq:intro-positive-euclidean}
\end{equation}
and interpret the integrand as a probability density for Monte Carlo sampling.
This positivity underlies the effectiveness of Euclidean lattice methods in
many bosonic theories and in fermionic theories for which the fermion
determinant is nonnegative.

At nonzero baryon chemical potential the same representation is generally no
longer positive.  In lattice QCD, for example, integrating out the fermions
gives
\begin{equation}
        Z(\mu)
        =
        \int DU\,
        \e^{-S_g[U]}
        \det D[U,\mu],
        \label{eq:intro-qcd-mu}
\end{equation}
and the determinant is generically complex for real baryon chemical potential.
The functional integral remains algebraically meaningful, but it no longer
defines a probability measure.  Reweighting with the determinant phase then
requires resolving cancellations whose severity commonly grows exponentially
with spacetime volume.  This is the finite-density sign problem; see
\cite{Nagata2022FiniteDensity,AlexandruBasarBedaqueWarrington2022,
AartsSexty2026}.

A Hamiltonian formulation avoids this particular complex-weight obstruction.
On a finite lattice, finite density enters through
\[
        H(\mu)=H-\mu N,
\]
which remains self-adjoint under the usual assumptions.  The difficulty is
instead the growth of the many-body Hilbert space.  Thus a fully Euclidean
description may lead to a sign or phase problem, whereas a fully Hamiltonian
description exposes the full state-space complexity.  This tension motivates
Hamiltonian lattice field theory and quantum simulation
\cite{BauerDavoudiEtAl2023}, as well as mixed classical formulations that use
Euclidean and Hamiltonian descriptions for different parts of the calculation.

The present paper studies one such mixed formulation.  The proposal is not an
algebraic solution of the sign problem.  It concerns the stage at which a
residual sector is converted into a scalar Euclidean weight.  Suppose that
\(X\) denotes the variables treated by a Monte Carlo-tractable reference
calculation, and let \(\Kcal\) be the Hilbert space of the remaining active
degrees of freedom.  For each configuration \(X\), the active sector contributes
an imaginary-time evolution operator
\[
        \mathcal U_{\rm active}[X]
        \in
        \operatorname{End}(\Kcal).
\]
The usual scalar route traces this sector out,
\begin{equation}
        \mathcal U_{\rm active}[X]
        \longmapsto
        \operatorname{Tr}_{\Kcal}\mathcal U_{\rm active}[X],
        \label{eq:intro-scalarization}
\end{equation}
or replaces the corresponding quadratic fermion trace by a determinant.  When
the resulting scalar is nonnegative, this is an exceptionally effective
reduction.  When it is sign-changing or complex, however, the same
scalarization produces the probability-sampling bottleneck.

Conditional Euclidean--Hamiltonian reduction, abbreviated CEH, stops before
this final trace.  The reference sector is treated by a Monte Carlo-tractable
Euclidean calculation, while the residual sector remains operator-valued and
is compressed to an \(R\)-dimensional trial space by
\(P_R\,\mathcal U_{\rm active}[X]P_R\).  In practice, one need not construct this
projected operator separately for every sampled configuration.  The primary
Euclidean outputs are finite correlation or transfer matrices between selected
trial states or interpolating operators.

A product Hilbert-space decomposition,
\[
        \Hcal=L^2(X)\otimes\Kcal,
        \qquad
        H=H_x+A(\hat x),
\]
provides the most transparent realization of this idea.  Resolving only the
\(x\)-coordinate in the Euclidean kernel leaves an operator on \(\Kcal\).
Tracing over \(\Kcal\) produces a scalar path weight, whereas retaining selected
matrix elements produces active-space data.  Section~\ref{sec:operator-valued-kernels}
reviews this standard operator-valued kernel construction and specializes it
to isolate the scalarization step used in CEH.

The computational principle is more general than a literal decomposition into
two tensor factors.  One may instead use an operator split
\begin{equation}
        H=G+A,
        \label{eq:intro-operator-split}
\end{equation}
where \(G\) defines a reference problem that can be treated efficiently by a
Euclidean Monte Carlo or related sampling method, while \(A\) is retained as an
operator-valued residual on the projected space.  The reference problem need
not be strictly sign-free; it is sufficient that any remaining sign or phase
fluctuations can be controlled at acceptable cost.  The split may therefore
separate interaction terms, charged sectors, clusters, or other selected
contributions rather than distinct physical subsystems.  In an
interaction-picture representation, the evolution generated by \(G\) supplies
the reference ensemble, while ordered insertions of \(A\) enter the
matrix-valued observables.

At finite density, one may distribute the chemical-potential term between the
two parts,
\[
        H(\mu)=G(\mu)+A(\mu),
\]
provided the reference problem remains tractable.  A particularly simple
choice is to leave the residual term \(-\mu N\) entirely for the projected
model.  More generally, part of the chemical-potential dependence may already
be included in \(G(\mu)\), with only the remaining contribution applied after
projection.  In either case, the difficult task is to construct and harvest a
projected space containing the number sectors relevant to the target density.

Choose trial states \(\Phi_I\), or interpolating operators that create them,
and let \(T=\e^{-aH}\).  The Euclidean input is the correlation matrix
\begin{equation}
        C_{IJ}(n)
        =
        \langle\Phi_I,T^n\Phi_J\rangle,
        \qquad
        I,J=1,\ldots,R.
        \label{eq:intro-correlation-matrix}
\end{equation}
These matrices determine the transfer dynamics resolved by the trial space,
either through a projected transfer operator or through a correlation-matrix
generalized eigenvalue problem.  In the simple case in which the full
chemical-potential term is deferred, the reduced finite-density Hamiltonian is
\[
        H_{\rm eff}(\mu)
        =
        H_{\rm eff}(0)-\mu N_{\rm eff}.
\]
If part of the chemical-potential dependence is already included in the
reference problem, only the remaining projected contribution is added at this
stage.  Thus the Euclidean GEVP extracts the projected dynamics, while the
deferred finite-density contribution is applied in the resulting reduced
model.
The detailed
constructions are given in \S\ref{sec:operator-valued-kernels} and \S\ref{sec:correlation-reduction}.

At finite rank, this prescription defines an effective model.  Direct Galerkin
projection commutes exactly with adding \(-\mu N\) when \(H\) and \(N\) are
projected on the same trial space.  A Hamiltonian reconstructed from a
compressed transfer operator is more generally an effective generator of the
resolved Euclidean dynamics.  Its accuracy must therefore be tested through
rank enlargement, Euclidean-time stability, metric conditioning,
number-sector coverage, and possible number-sector leakage.

CEH therefore changes the computational difficulty rather than eliminating it.
For the trade to be useful, two conditions must hold.  First, the residual
active dynamics must be low-rank approximable for the observables of interest:
a moderate-dimensional trial space must give stable results under enlargement.
Second, the required matrix data must be reference-harvestable: the chosen reference calculation must resolve the relevant active or charged directions
with usable overlap, variance, operator-basis size, and metric conditioning.
These requirements are distinct.  A small active sector may still be
effectively invisible to the chosen reference ensemble, while accurately
measured correlators are not useful if the required rank approaches the full
Hilbert-space dimension.

The relevant failure modes are consequently explicit.  The active rank may
grow too rapidly; small metric directions may be lost to noise; the operator
dictionary may fail to overlap with the relevant charged states; or the
finite-density tilt may select number sectors absent from the projected space.
In such regimes CEH degenerates either toward full Hamiltonian simulation or
toward an ill-conditioned signal-extraction problem.  No general claim about
volume scaling or asymptotic complexity is made in this paper.

The paper formulates the pre-scalarization split in operator-valued-kernel and
correlation-matrix language, relates projected transfer matrices to conditional
Galerkin reduction, and tests the resulting diagnostics in three finite
benchmarks: conditional compression and stochastic GEVP harvesting in a
two-channel oscillator, followed by a determinant-sign and non-target
sector-coverage study in a four-site Hubbard ring.

Most of the ingredients used here are standard.  Operator-valued Euclidean
kernels and semi-Euclidean constructions have mathematical precedents
\cite{HershPapanicolaou1972,BrydgesFederbush1974}.  Correlation-matrix GEVPs
are established tools in lattice spectroscopy
\cite{LuscherWolff1990,BlossierEtAl2009}, and the Monte Carlo Hamiltonian
program already constructs effective Hamiltonians from Euclidean transition
amplitudes \cite{JirariKrogerLuoMoriarty1999}.  CEH does not claim novelty for
operator-valued path integrals, Rayleigh--Ritz projection, GEVPs, or the
reconstruction of finite Hamiltonians from Euclidean data.

The closest recent projected-matrix relative is projected density matrix
sampling (PDMS) \cite{KarnaWuChandrasekharanKaul2025}.  PDMS estimates
projected thermal density and energy matrices of the full Hamiltonian; if the
chosen path-integral representation is signful, the signs enter the measured
observables.  
CEH shares the final projected-matrix linear algebra but changes
the upstream construction: a tractable reference formulation is chosen before
the residual sector is scalarized, and conditioned active-sector matrix data
are measured instead.
Other recent hybrid methods combining Monte Carlo with Hamiltonian solvers,
together with other matrix-valued sampling approaches, further show that
neither ingredient is by itself new
\cite{TemmenGistiLuitzLuuOstmeyer2026,
BollmarkMardazadHofmannKantian2024,WangWangMaoYan2026GRDM}.
The specific point examined here is the placement of the split before
scalarization, together with finite-density continuation on the resulting
projected model.

The scope is therefore deliberately limited.  This paper asks whether the CEH
ordering can replace scalar sign reweighting by a monitored active-space
approximation in controlled finite examples.  It does not claim a generic
solution of the sign problem, a new projected-Hamiltonian formalism, or a
theorem guaranteeing active-sector compressibility or reference harvestability.

Section~\ref{sec:operator-valued-kernels} reviews the operator-valued kernel
and its scalar-weight counterpart and isolates the scalarization step relevant
to CEH.
Section~\ref{sec:correlation-reduction} gives the correlation-matrix,
projected-transfer-matrix, conditional Galerkin, and finite-density
constructions.
Section~\ref{sec:numerical-results} presents the oscillator and Hubbard
experiments.
Section~\ref{sec:conclusion} summarizes the conclusions and remaining
limitations.

\section{Operator-valued kernels and scalar weights}
\label{sec:operator-valued-kernels}

This section isolates the step at which a residual quantum sector becomes a
scalar Euclidean weight.  A product Hilbert-space model gives the clearest
formulation: resolving only the reference coordinate produces an
operator-valued kernel, while tracing the remaining sector produces a scalar
trace or determinant.  When this scalar is nonnegative, the trace is precisely
what makes Euclidean Monte Carlo effective.  When it is sign-changing or
complex, the same scalarization creates the sign or phase problem.  CEH instead
uses the reference calculation to obtain finite matrix data before that trace
is taken.

\subsection{Sector and operator splits}
\label{subsec:model-class}

Consider first the product Hilbert space
\begin{equation}
        \Hcal=L^2(X,dx)\otimes\Kcal ,
        \label{eq:product-H}
\end{equation}
where \(X\) is the configuration space of the reference sector and
\(\Kcal\) is the active Hilbert space.  Let
\begin{equation}
        H=H_x+A(\hat x),
        \label{eq:H-basic}
\end{equation}
where \(H_x\) acts on \(L^2(X)\) and \(A(x)\) acts on \(\Kcal\).
The reference Hamiltonian \(H_x\) is assumed to admit an efficiently
tractable Euclidean representation.  The cleanest case is a nonnegative
sampling measure, although a controlled residual sign or phase problem is
also compatible with the general construction.  A typical reference
Hamiltonian is
\begin{equation}
        H_x
        =
        -\frac12\nabla_x\cdot M^{-1}\nabla_x+V(x).
        \label{eq:Hx-basic}
\end{equation}
We assume that \(H_x\) is self-adjoint and bounded below and that
\(H_x+A(\hat x)\) defines a self-adjoint, bounded-below form sum.  This includes
the oscillator examples below: their linear couplings are relatively
form-bounded with respect to the confining quadratic Hamiltonian.  The
finite-slice expressions should be read literally when \(A\) is bounded and as
the corresponding Trotter approximations in the unbounded examples.

A state in \(\Hcal\) can be viewed as a \(\Kcal\)-valued function
\(\Psi(x)\), with
\begin{equation}
        \ip{\Psi}{\Phi}_{\Hcal}
        =
        \int_X
        \ip{\Psi(x)}{\Phi(x)}_{\Kcal}\,dx .
        \label{eq:K-valued-inner-product}
\end{equation}
The variable \(x\) is treated through the reference Euclidean representation,
while the value \(\Psi(x)\in\Kcal\) retains the active degrees of freedom.

The product decomposition is a concrete realization rather than a restriction
on the general architecture.  More generally, one may begin with an operator
split
\begin{equation}
        H=G+A,
        \label{eq:S2-operator-split}
\end{equation}
where \(G\) admits a positive or otherwise controlled Euclidean representation
and \(A\) is retained as an operator-valued residual.  Formally,
\begin{equation}
        \e^{-\beta(G+A)}
        =
        \e^{-\beta G}
        \mathcal T_\tau
        \exp\left(
        -\int_0^\beta A_I(\tau)\,d\tau
        \right),
        \qquad
        A_I(\tau)=\e^{\tau G}A\e^{-\tau G}.
        \label{eq:S2-interaction-picture}
\end{equation}
The reference evolution generated by \(G\) supplies the Euclidean data, while
the ordered residual insertions are kept inside finite matrix observables.
Such a split may separate terms, number sectors, clusters, or other selected
degrees of freedom rather than literal physical subsystems.  

When \(H(0)\) supplies a tractable reference problem, assigning the full
chemical-potential term to the residual part gives the operator split
\(G=H(0)\), \(A=-\mu N\), so that \(H(\mu)=G+A\).
The product model \eqref{eq:H-basic} will be used below
because it displays the scalarization step explicitly; the
correlation-matrix construction of Section~\ref{sec:correlation-reduction}
applies equally to the broader operator-split interpretation.

\subsection{The operator-valued Euclidean kernel}
\label{subsec:operator-valued-kernel}

For \(\beta>0\), resolve only the \(x\)-coordinate in the Euclidean kernel:
\begin{equation}
        K_\beta(x_f,x_i)
        =
        \langle x_f|\e^{-\beta H}|x_i\rangle_x .
        \label{eq:operator-kernel-def}
\end{equation}
This is an operator on \(\Kcal\), characterized by
\begin{equation}
        \bigl(\e^{-\beta H}\Phi\bigr)(x_f)
        =
        \int_X K_\beta(x_f,x_i)\Phi(x_i)\,dx_i .
        \label{eq:kernel-action}
\end{equation}
Equivalently,
\begin{equation}
        \ip{\eta_f}{K_\beta(x_f,x_i)\eta_i}_{\Kcal}
        =
        \langle x_f,\eta_f|
        \e^{-\beta H}
        |x_i,\eta_i\rangle_{\Hcal},
        \qquad
        \eta_i,\eta_f\in\Kcal .
        \label{eq:kernel-matrix-elements}
\end{equation}
Thus \(K_\beta(x_f,x_i)\) is simply the heat kernel of \(H\) with the reference coordinate resolved and the active sector left untraced.

Set \(\beta=Na\) and write
\[
        k_a(x',x)
        =
        \langle x'|\e^{-aH_x}|x\rangle .
\]
Inserting resolutions of identity only in the \(x\)-space gives the
finite-slice operator-valued kernel
\begin{equation}
\begin{aligned}
        K_{\beta,a}(x_N,x_0)
        =
        \int_{X^{N-1}}
        \prod_{n=1}^{N-1}dx_n\,
        \prod_{n=0}^{N-1}k_a(x_{n+1},x_n)\,
        \e^{-aA(x_{N-1})}\cdots\e^{-aA(x_0)},
        \label{eq:operator-valued-discrete}
\end{aligned}
\end{equation}
where \(x_0=x_i\) and \(x_N=x_f\).  A symmetric splitting may be used in place
of the displayed first-order form.  The ordering of the active factors is
essential because the matrices \(A(x_n)\) need not commute.

The corresponding continuum notation is
\begin{equation}
        K_\beta(x_f,x_i)
        =
        \int_{x(0)=x_i}^{x(\beta)=x_f}
        \e^{-S_x[x]}\,
        \mathcal T_\tau
        \exp\left(
        -\int_0^\beta A(x(\tau))\,d\tau
        \right)
        Dx .
        \label{eq:operator-valued-continuum}
\end{equation}
Equation~\eqref{eq:operator-valued-discrete}, or its symmetric-splitting
analogue, is the precise finite-slice object used whenever discretization
details matter.  
The continuum expression records its structure: each reference-sector history carries an ordered operator on \(\Kcal\).

\subsection{Scalarization}
\label{subsec:scalar-weight-formulation}

For periodic Euclidean time, the finite-slice partition function can be written
as
\begin{equation}
\begin{aligned}
        Z_a(\beta)
        &=
        \int_{X^N}
        \prod_{n=0}^{N-1}dx_n\,
        \prod_{n=0}^{N-1}k_a(x_{n+1},x_n)\,
        W_a[x_0,\ldots,x_{N-1}],
        \qquad x_N=x_0,
        \\
        W_a[x_0,\ldots,x_{N-1}]
        &=
        \Tr_{\Kcal}
        \left[
        \e^{-aA(x_{N-1})}\cdots\e^{-aA(x_0)}
        \right].
        \label{eq:scalar-weight}
\end{aligned}
\end{equation}
The map from the ordered active-sector operator to \(W_a\) is the
scalarization step.

When \(W_a\) is real and nonnegative, the integrand defines an ordinary
Euclidean probability measure.  This is the successful regime of determinant,
auxiliary-field, and related Monte Carlo methods.  Scalarization has then
compressed a potentially large quantum trace into a compact positive weight.

If instead
\(W_a=|W_a|\e^{\ii\theta_a}\),
one may sample the measure proportional to \(|W_a|\) and restore the phase by
reweighting.  The resulting denominator 
\(\langle\e^{\ii\theta_a}\rangle_{\rm pq}\)
may become exponentially small in spacetime volume or inverse temperature.
The difficulty is not an inaccurate pointwise evaluation of \(W_a\); it is
that a sign-changing or complex weight does not define the nonnegative
probability measure required for ordinary importance sampling, so the physical
result must be recovered from cancellations.

For a quadratic number-conserving fermionic sector, the same scalarization is
usually expressed as a determinant.  If \(U_f[x]\) is the one-particle
imaginary-time evolution along the history \(x\), and \(\Gamma(U_f[x])\) is
its second quantization on the fermionic Fock space, then
\begin{equation}
        \Tr_{\Fcal}\Gamma(U_f[x])
        =
        \det\bigl(I+U_f[x]\bigr).
        \label{eq:fermion-trace-determinant}
\end{equation}
This identity is one of the principal computational advantages of fermionic
Monte Carlo.  It becomes a sign-problem bottleneck only when the resulting
determinant remains sign-changing or complex.

\subsection{Leaving the active sector untraced}
\label{subsec:untraced-active-sector}

In the positive-measure realization used below, CEH follows the other route
through \eqref{eq:operator-valued-discrete}.  Choose a finite orthonormal family
\[
        \{|\alpha\rangle\}_{\alpha=1}^R
        \subset\Kcal
\]
and let \(P_R\) denote the orthogonal projection onto its span.  For a fixed
reference-sector history \(x\), define the compressed active-sector evolution
\begin{equation}
        M_a^{(R)}[x]
        =
        \left.
        P_R
        \e^{-aA(x_{N-1})}\cdots\e^{-aA(x_0)}
        P_R
        \right|_{\operatorname{Ran}P_R}
        \in\operatorname{End}(\mathbb C^R).
        \label{eq:matrix-valued-estimator}
\end{equation}
Its entries are
\[
        \bigl(M_a^{(R)}[x]\bigr)_{\alpha\beta}
        =
        \langle\alpha|
        \e^{-aA(x_{N-1})}\cdots\e^{-aA(x_0)}
        |\beta\rangle .
\]

Let \(G_a[x]\geq0\) denote the reference-sector path weight and let
\(\mathbb E_{\rm ref}\) denote expectation with respect to the corresponding
normalized path measure.
After fixing or integrating the endpoint data and absorbing the corresponding
bridge normalization into \(\mathcal N_{\rm ref}\), the projected partial
kernel is obtained schematically as
\begin{equation}
        K^{(R)}_{\alpha\beta}(\beta)
        =
        \mathcal N_{\rm ref}(\beta)\,
        \mathbb E_{\rm ref}
        \left[
        \bigl(M_a^{(R)}[x]\bigr)_{\alpha\beta}
        \right] .
        \label{eq:direct-matrix-MC}
\end{equation}
In the continuum notation,
\(M_a^{(R)}[x]\) is replaced by the compression of the ordered exponential
\[
        \mathcal T_\tau
        \exp\left[
        -\int_0^\beta A(x(\tau))\,d\tau
        \right].
\]
Thus each sampled reference history contributes a finite matrix, and averaging
these matrices produces the projected operator-valued kernel.

In a fixed active basis, the same untraced description is simply a
multicomponent Hamiltonian problem.  Writing
\[
        \Psi(x)
        =
        \sum_\alpha\psi_\alpha(x)|\alpha\rangle
\]
gives
\begin{equation}
        H_x\psi_\alpha(x)
        +
        \sum_\beta
        A_{\alpha\beta}(x)\psi_\beta(x)
        =
        E\psi_\alpha(x).
        \label{eq:coupled-channel-basic}
\end{equation}
The conditional Galerkin construction of Section~\ref{sec:correlation-reduction}
is a finite-dimensional Rayleigh--Ritz realization of this coupled problem.

Equation~\eqref{eq:direct-matrix-MC} is mainly conceptual.  Direct entrywise
sampling may become expensive or noisy as the active dimension grows.  In
lattice and many-body applications, the natural Euclidean outputs are instead
finite correlation matrices between trial states or interpolating operators.
Those matrices, rather than the full configuration-space kernel, are the
starting point of Section~\ref{sec:correlation-reduction}.

\subsection{What is gained and what is not}
\label{subsec:what-is-gained}

The distinction is now explicit.  For a path \(x\), let
\[
        \mathcal U_a[x]
        =
        \e^{-aA(x_{N-1})}\cdots\e^{-aA(x_0)}
\]
denote the full conditioned active evolution.  The scalar route retains only
\(\Tr_{\Kcal}\mathcal U_a[x]\), whereas CEH retains the compressed matrix
\(M_a^{(R)}[x]\) defined in \eqref{eq:matrix-valued-estimator} and averages its
entries.  Thus the scalar active-sector trace is replaced by matrix-valued
harvesting.

This replaces scalar sign or phase reweighting by a finite-rank matrix
problem.  It is useful only when the residual dynamics are low-rank
approximable and the required matrix data are reference-harvestable.  The split
must therefore balance reference tractability against rank, overlap, variance,
and conditioning; otherwise CEH degenerates toward full Hamiltonian simulation
or ill-conditioned signal extraction.  Section~\ref{sec:correlation-reduction}
turns these failure modes into explicit diagnostics.

\section{Correlation-matrix Euclidean--Hamiltonian reduction}
\label{sec:correlation-reduction}

The operator-valued kernel of
Section~\ref{sec:operator-valued-kernels} identifies the scalarization step.
For computation, however, one usually does not tabulate kernels between
configurations or store wavefunctions over the reference-sector variables.  The
natural Euclidean output is a finite family of correlation matrices between
chosen trial states or interpolating operators.  These matrices define the
resolved transfer dynamics and, after projection, a finite Hamiltonian model.

Let
\(T=\e^{-aH}\)
be the Euclidean one-step transfer operator, and choose trial states
\[
        \Phi_I\in\Hcal,
        \qquad
        I=1,\ldots,R.
\]
The trial states are chosen to carry the active-sector quantum numbers and
number sectors relevant to the target observables.  In a finite model they may
be explicit vectors.  In a lattice-field-theory setting they are more naturally
created by interpolating operators acting on a reference state or ensemble.

\subsection{Correlation matrices as the primary Euclidean input}
\label{subsec:correlation-primary}

Define
\begin{equation}
        C_{IJ}(n)
        =
        \ip{\Phi_I}{T^n\Phi_J}_{\Hcal},
        \qquad
        n=0,1,2,\ldots .
        \label{eq:correlation-matrix}
\end{equation}
At \(n=0\), \(C(0)\) is the Gram matrix of the trial states.  For \(n>0\),
\(C(n)\) records Euclidean propagation between the same states.  The central
computational object is therefore the matrix family \(C(n)\), rather than a
scalar determinant or trace.

In the positive-reference specialization, each entry can be viewed
schematically as a matrix-valued observable averaged over reference-sector
histories.  If \(x\) denotes a sampled history and
\(U_{\Kcal}^{(n)}[x]\) is the conditioned active-sector evolution over time
\(na\), then
\begin{equation}
        C_{IJ}(n)
        =
        \mathbb E_{\rm ref}
        \left[
        \left\langle
        v_I^{\rm out}[x],
        U_{\Kcal}^{(n)}[x]v_J^{\rm in}[x]
        \right\rangle_{\Kcal}
        \right].
        \label{eq:correlation-reference-expectation}
\end{equation}
The source and sink vectors depend on the chosen operators and boundary
conditions.  The structural point is independent of these details: the
reference variables define the positive sampling measure, while the residual
sector appears inside finite-dimensional matrix observables.  One does not use
\(\Tr_{\Kcal}U_{\Kcal}^{(n)}[x]\) as part of the probability weight.

The same formulation applies to the operator split \(H=G+A\), cf. \eqref{eq:intro-operator-split} and \eqref{eq:S2-operator-split}.  The reference evolution generated
by \(G\) supplies the Euclidean ensemble, and ordered residual insertions of
\(A\) enter the trial-space correlation matrices.  A literal tensor-product
decomposition is therefore not required once the finite trial states and their
correlators have been specified.

\subsection{Projected transfer matrices and the GEVP}
\label{subsec:projected-transfer-gevp}

Assume first that \(C(0)\) is positive definite.  Orthogonalizing the trial
states gives the projected transfer matrix
\begin{equation}
        T_{\rm eff}
        =
        C(0)^{-1/2}C(1)C(0)^{-1/2}.
        \label{eq:Teff-C0C1}
\end{equation}
More generally, one may use a nonzero reference time \(n_0\):
\begin{equation}
        T_{\rm eff}(n_0)
        =
        C(n_0)^{-1/2}
        C(n_0+1)
        C(n_0)^{-1/2}.
        \label{eq:Teff-n0}
\end{equation}
This is the transfer operator compressed to the Euclidean-filtered trial
space generated by \(T^{n_0/2}\Phi_I\).  Increasing \(n_0\) suppresses
high-energy contamination, but it can also remove weakly resolved directions
and worsen the conditioning of \(C(n_0)\).

The projected energies are usually extracted through the generalized
eigenvalue problem
\begin{equation}
        C(n)v_k
        =
        \lambda_k(n,n_0)\,C(n_0)v_k,
        \label{eq:GEVP-main}
\end{equation}
with
\begin{equation}
        E_k^{\rm eff}(n,n_0)
        =
        -\frac{1}{a(n-n_0)}
        \log\lambda_k(n,n_0).
        \label{eq:effective-energy}
\end{equation}
If the trial space is invariant under the transfer operator, these quantities
recover the corresponding exact energies and are independent of \(n\) and
\(n_0\).  In a finite noninvariant space, one instead looks for stability under
changes of the trial basis and Euclidean-time window.

When \(T_{\rm eff}\) is positive definite and its eigenvalues are separated
from zero,
one may define
\begin{equation}
        H_{\rm eff}
        =
        -a^{-1}\log T_{\rm eff}.
        \label{eq:Heff-log-main}
\end{equation}
In stochastic calculations, solving the GEVP is often more stable than forming
the logarithm of a noisy matrix.  We nevertheless use \(H_{\rm eff}\) as a
convenient name for the finite dynamics represented by the projected transfer
matrix or its GEVP spectrum.

\begin{proposition}[Projected transfer matrix]
\label{prop:projected-transfer}
Let \(H\) be self-adjoint and bounded below, let \(T=\e^{-aH}\), and let
\(
        \Vcal_R
        =
        \operatorname{span}\{\Phi_1,\ldots,\Phi_R\}.
\)
Assume that \(C(0)>0\), and define \(J_R:\mathbb C^R\to\Hcal\) by
\[
        J_Rc=\sum_{I=1}^R c_I\Phi_I.
\]
Then we have
\[
C(0)=J_R^*J_R,\qquad C(1)=J_R^*TJ_R,
\]
and
\(U_R=J_RC(0)^{-1/2}\)
is an isometry from \(\mathbb C^R\) onto \(\Vcal_R\).  Consequently,
\begin{equation}
        T_{\rm eff}
        =
        C(0)^{-1/2}C(1)C(0)^{-1/2}
        =
        U_R^*TU_R
        \label{eq:Teff-compression}
\end{equation}
is the compression of the transfer operator to \(\Vcal_R\), written in an
orthonormal basis.  Its eigenvalues are Ritz values of \(T\) on \(\Vcal_R\).  
Under the usual Rayleigh--Ritz approximation assumptions,
enlargement of the trial spaces gives convergence of the corresponding
isolated low-energy levels.
\end{proposition}

\begin{proof}
The first two identities follow immediately from
\eqref{eq:correlation-matrix}.  Since \(C(0)>0\), we have
\[
        U_R^*U_R
        =
        C(0)^{-1/2}J_R^*J_RC(0)^{-1/2}
        =
        I,
\]
and the range of \(U_R\) is \(\Vcal_R\).  Hence \(U_R^*TU_R\) is precisely the
compression of \(T\) to the trial space.  The final statement is the standard
Rayleigh--Ritz convergence result for the bounded self-adjoint operator \(T\),
combined with the spectral relation \(T=\e^{-aH}\).
\end{proof}

The same argument applies at nonzero reference time \(n_0\) after replacing \(J_R\) by \(T^{n_0/2}J_R\).  
Thus the metric \(C(n_0)\) describes the linearly
independent Euclidean-filtered trial directions that remain resolved at that
time.

\subsection{Conditional Galerkin spaces}
\label{subsec:conditional-galerkin-interpretation}

The correlation-matrix construction is closest to Euclidean Monte Carlo.  A
conditional Galerkin basis gives the corresponding finite-dimensional
Hilbert-space picture and is useful in benchmark tests.

Choose an orthonormal active basis
\(
        \{|\alpha\rangle\}_{\alpha\in\mathcal A}
        \subset\Kcal .
\)
Writing
\[
        \Psi(x)
        =
        \sum_{\alpha\in\mathcal A}
        \psi_\alpha(x)|\alpha\rangle
\]
turns \(H\Psi=E\Psi\) into the exact coupled-channel system
\begin{equation}
        H_x\psi_\alpha(x)
        +
        \sum_{\beta\in\mathcal A}
        A_{\alpha\beta}(x)\psi_\beta(x)
        =
        E\psi_\alpha(x).
        \label{eq:coupled-channel-S3}
\end{equation}

A channel-adapted basis can be constructed from the diagonal Hamiltonians
\begin{equation}
        H_x^{(\alpha)}
        =
        H_x+A_{\alpha\alpha}(x).
        \label{eq:channel-H-S3}
\end{equation}
Let
\begin{equation}
        H_x^{(\alpha)}\chi_r^{(\alpha)}
        =
        \epsilon_r^{(\alpha)}\chi_r^{(\alpha)}
        \label{eq:conditional-basis-S3}
\end{equation}
in the discrete-spectrum case, or let
\(\{\chi_r^{(\alpha)}\}\) be any finite channel-adapted family obtained by a
deterministic or Euclidean projection procedure.  The trial space is
\begin{equation}
        \Vcal_R
        =
        \operatorname{span}
        \left\{
        \chi_r^{(\alpha)}|\alpha\rangle:
        \alpha\in\mathcal A_R,\ 
        0\leq r<R_\alpha
        \right\}.
        \label{eq:VR-conditional}
\end{equation}
The overlap and Hamiltonian matrices are
\begin{align}
        (S_R)_{\alpha r,\beta s}
        &=
        \left\langle
        \chi_r^{(\alpha)}|\alpha\rangle,
        \chi_s^{(\beta)}|\beta\rangle
        \right\rangle_{\Hcal},
        \label{eq:SR-general-S3}
        \\
        (H_R)_{\alpha r,\beta s}
        &=
        \left\langle
        \chi_r^{(\alpha)}
        \middle|
        H_x\delta_{\alpha\beta}+A_{\alpha\beta}(x)
        \middle|
        \chi_s^{(\beta)}
        \right\rangle_{L^2(X)} .
        \label{eq:HR-general-S3}
\end{align}
The reduced problem is
\begin{equation}
        H_Rc=E_RS_Rc.
        \label{eq:generalized-HR-S3}
\end{equation}
It is the ordinary Rayleigh--Ritz problem for \(H\) on \(\Vcal_R\).  In
particular, when the trial vectors lie in the form domain of \(H\), the
ground-state Ritz value is an upper approximation to the exact ground-state
energy and converges under the standard trial-space enlargement assumptions.

If the active labels are orthogonal and the conditional functions are
orthonormal within each channel, then \(S_R=I\).  Cross-channel overlaps may
nevertheless enter \(H_R\) through the off-diagonal couplings
\(A_{\alpha\beta}(x)\).  This mechanism will be used in the two-channel
oscillator test of \S\ref{subsec:numerics-e1}: the Gram matrix is trivial,
while the spin-flip block contains overlaps between oppositely displaced
oscillator states.

The construction uses a fixed active basis.  If the active basis itself is
allowed to depend on \(x\), the corresponding derivative or connection terms
must be included in the projected Hamiltonian.

\subsection{Finite density after projection}
\label{subsec:finite-density-after-projection}

Let \(N\) be a conserved number operator, with \([H,N]=0\).  In the full
Hamiltonian description, we have
\[
        H(\mu)=H-\mu N.
\]
This subsection treats the simple case in which the full chemical-potential
term is deferred until after the finite trial space has been selected or
resolved from Euclidean data.  If part of the chemical-potential dependence is
included in the reference problem, the same discussion applies to the
remaining projected contribution.  At finite rank, the interpretation depends
on whether one uses a direct Hamiltonian compression or a transfer-matrix
reconstruction.

For a direct Galerkin projection, let
\(U_R:\mathbb C^R\to\Hcal\) be an isometry with range \(\Vcal_R\), and define
\[
        H_R=U_R^*HU_R,
        \qquad
        N_R=U_R^*NU_R.
\]
We then infer
\begin{equation}
        H_R-\mu N_R
        =
        U_R^*(H-\mu N)U_R.
        \label{eq:direct-galerkin-mu}
\end{equation}
Thus the finite-density Hamiltonian is exactly the Galerkin compression of
\(H-\mu N\) on the same trial space.

Number conservation inside the projected model requires more than the identity
\eqref{eq:direct-galerkin-mu}.  If \(\Vcal_R\) is invariant under \(N\), its
number sectors are represented exactly and
\[
        [H_R,N_R]=0.
\]
If it is not invariant, \(N_R\) remains self-adjoint, but projection can give
\begin{equation}
        [H_R,N_R]\neq0
        \qquad\text{even though}\qquad
        [H,N]=0.
        \label{eq:number-leakage-S3}
\end{equation}
The commutator therefore provides a diagnostic of number-sector leakage.  A
small commutator alone is not sufficient: the trial space must also contain the
number sectors whose tilted levels \(E_{N,j}-\mu N\) contribute to the target
observable.

The correlation-matrix route contains an additional approximation.  Exact
correlation data determine a compressed transfer operator
\[
        T_R=U_R^*\e^{-aH}U_R,
\]
and, when stable,
\[
        H_{\rm eff}(0)=-a^{-1}\log T_R .
\]
Unless \(\Vcal_R\) is invariant under \(\e^{-aH}\),
\begin{equation}
        -a^{-1}
        \log\!\left(U_R^*\e^{-aH}U_R\right)
        \neq
        U_R^*HU_R
        \label{eq:log-compression-S3}
\end{equation}
in general.  Thus the finite-density prescription
\begin{equation}
H_{\rm eff}(\mu) = H_{\rm eff}(0)-\mu N_{\rm eff}
\label{eq:Heff-mu-S3}
\end{equation}
defines an effective projected model, not necessarily the direct Galerkin
compression of the full finite-density Hamiltonian.

At reference time \(n_0=0\), the number matrix in the orthonormalized trial
basis is
\begin{equation}
        N_{\rm eff}
        =
        C(0)^{-1/2}N_C C(0)^{-1/2},
        \qquad
        (N_C)_{IJ}
        =
        \ip{\Phi_I}{N\Phi_J}_{\Hcal}.
        \label{eq:Neff-orthonormal}
\end{equation}
The analogous construction may be made in a Euclidean-filtered trial space.
Thermal observables are then computed from
\begin{equation}
        Z_R(\beta,\mu)
        =
        \Tr_{\Vcal_R}
        \e^{-\beta(H_{\rm eff}(0)-\mu N_{\rm eff})} .
        \label{eq:ZR-mu-S3}
\end{equation}

This is the operator-split interpretation of finite density in its simplest
form.  The residual term \(-\mu N\) is inexpensive once the relevant canonical
sectors have been represented.  The nontrivial task is to harvest enough
sector-resolved information from the chosen reference calculation that the
finite-density tilt does not select states absent from the projected space.

\subsection{Conditioning and stability diagnostics}
\label{subsec:conditioning-stability}

The correlation-matrix formulation exposes several distinct error sources.
The metric \(C(n_0)\) must be positive definite on the resolved trial space.
In exact arithmetic, a zero eigenvalue signals linear dependence among the
Euclidean-filtered trial states.  With stochastic data, small positive
eigenvalues can be obscured by matrix-element noise.

In practice, the measured matrices are symmetrized and the spectrum of
\(C(n_0)\) is inspected.  Directions below a specified eigenvalue threshold
are removed or regularized, and the resulting retained rank is reported.  The
projected spectrum should then be checked under changes of \(n_0\), \(n\), the
trial basis, the threshold, and the stochastic resample.  More elaborate
statistical treatments of noisy norm and Hamiltonian matrices are also
available; see, e.g., \cite{HicksLee2023Trimmed}.

It is useful to distinguish four approximation layers:

\begin{enumerate}
        \item finite-rank error from the choice of trial space;
        \item stochastic error in the entries of \(C(n)\);
        \item regularization error from removing poorly resolved metric
        directions;
        \item transfer-logarithm error or instability when retained transfer
        eigenvalues approach zero.
\end{enumerate}

The GEVP avoids forming a noisy matrix logarithm, but it does not remove the
metric-conditioning problem.  A stable calculation should report the resolved
rank, the smallest retained metric eigenvalue or condition number, and the
dependence of the target quantities on trial-space and Euclidean-time
enlargement.

In a direct Galerkin calculation one may additionally monitor the residual
\begin{equation}
        r_R
        =
        H\Psi_R-E_R\Psi_R.
        \label{eq:residual-S3}
\end{equation}
In a correlation-matrix calculation, analogous closure information can be
obtained by adding trial operators or measuring couplings to candidate external
directions.  At finite density, the corresponding diagnostics are
number-sector coverage and, where number conservation is expected,
\(\|[H_{\rm eff},N_{\rm eff}]\|\).

The two implementation routes are complementary: conditional Galerkin forms
\(S_R\), \(H_R\), and \(N_R\) directly, whereas the correlation-matrix route
estimates \(C_{IJ}(n)\), resolves the transfer dynamics, and applies the
deferred finite-density term.  Both require closure at moderate rank and usable
Euclidean visibility of the retained directions.

\section{Numerical experiments: compression, matrix harvesting, and
finite-density continuation}
\label{sec:numerical-results}

This section tests three components of CEH in finite models with exact or
independently controlled references: conditional basis compression, stochastic
harvesting of projected correlation matrices, and finite-density continuation
in a signful Hubbard benchmark.  The diagnostics are spectral and density
errors, residuals, matrix uncertainty, metric conditioning, retained rank,
number-sector coverage, and number-sector leakage.  These experiments test
finite-model control, not asymptotic scaling or production-code performance.

\subsection{Conditional compression in a two-channel oscillator}
\label{subsec:numerics-e1}

Consider the two-channel oscillator
\begin{equation}
        H
        =
        \frac{p^2}{2}
        +
        \frac{\Omega^2x^2}{2}
        +
        \epsilon\sigma_z
        +
        t\sigma_x
        +
        gx\sigma_z .
        \label{eq:spin-boson-test-numerics}
\end{equation}
The oscillator is the reference sector, and the two-dimensional spin space is
the active sector.
In the \(\sigma_z\) basis, the diagonal channel
Hamiltonians are the oppositely displaced oscillators
\begin{equation}
\begin{aligned}
        H_b^{(+)}
        &=
        \frac{p^2}{2}
        +
        \frac{\Omega^2x^2}{2}
        +
        gx+\epsilon,
        \\
        H_b^{(-)}
        &=
        \frac{p^2}{2}
        +
        \frac{\Omega^2x^2}{2}
        -
        gx-\epsilon .
\end{aligned}
\label{eq:e1-shifted-channels}
\end{equation}
Let
\[
        H_b^{(\pm)}\chi_r^{(\pm)}
        =
        \varepsilon_r^{(\pm)}\chi_r^{(\pm)} .
\]
The conditional Galerkin space with \(R\) oscillator states in each channel is
\begin{equation}
        \Vcal_R^{\rm cond}
        =
        \operatorname{span}
        \left\{
        \chi_r^{(+)}|+\rangle,\,
        \chi_r^{(-)}|-\rangle:
        0\leq r<R
        \right\}.
        \label{eq:spin-boson-conditional-basis-numerics}
\end{equation}
Because the spin labels are orthogonal and the oscillator states are
orthonormal within each channel,
\begin{equation}
        S_R^{\rm cond}=I,
        \qquad
        \kappa(S_R^{\rm cond})=1.
        \label{eq:e1-conditional-metric}
\end{equation}
The displaced-oscillator overlaps enter through the off-diagonal spin-flip
matrix elements
\begin{equation}
        \langle
        \chi_r^{(+)}|+\rangle,
        H\chi_s^{(-)}|-\rangle
        \rangle
        =
        t\langle\chi_r^{(+)}|\chi_s^{(-)}\rangle .
        \label{eq:e1-franck-condon}
\end{equation}

For comparison, we use the equal-dimensional undressed product space
\begin{equation}
        \Vcal_R^{\rm prod}
        =
        \operatorname{span}
        \left\{
        \varphi_r|+\rangle,\,
        \varphi_r|-\rangle:
        0\leq r<R
        \right\},
        \label{eq:e1-product-space}
\end{equation}
where \(\varphi_r\) are eigenfunctions of the undisplaced oscillator
\[
        H_b=\frac{p^2}{2}+\frac{\Omega^2x^2}{2}.
\]
Both spaces have total dimension \(2R\).  The conditional basis incorporates
the channel displacement before the channels are coupled, whereas the product
basis must reproduce the same displaced states through linear combinations of
undisplaced oscillator levels.

For the compression test we use
\[
        \Omega=1,
        \qquad
        \epsilon=0.3,
        \qquad
        t=0.2,
        \qquad
        g=2,
\]
with an exact-diagonalization reference converged under further enlargement of
the oscillator cutoff.  Figure~\ref{fig:e1-conditional-vs-full} compares the
two equal-dimensional trial spaces.  At \(R=2,4,8\), the conditional
ground-state errors are respectively
\(5.23\times10^{-3}\), \(4.74\times10^{-3}\), and
\(2.11\times10^{-3}\), compared with
\(9.88\times10^{-1}\), \(2.48\times10^{-1}\), and
\(4.75\times10^{-3}\) for the undressed product basis.  The advantage is
therefore largest at very small rank and decreases once the product basis
contains enough oscillator levels to resolve the displaced states.

\begin{figure}[ht]
        \centering
        \includegraphics[width=0.7\textwidth]
        {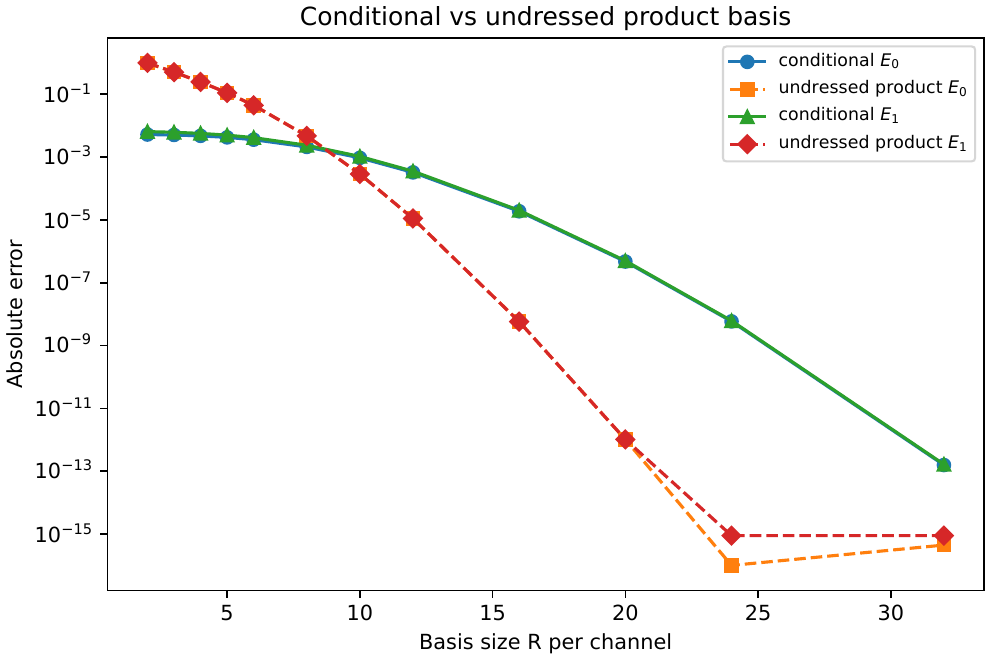}
        \caption{Low-energy errors versus \(R\) per channel.}
        \label{fig:e1-conditional-vs-full}
\end{figure}

For the conditional ground state, the residual
\(\|H\Psi_R-E_R\Psi_R\|\) decreases from approximately
\(2.00\times10^{-1}\) at \(R=2\) to \(1.81\times10^{-2}\) at \(R=16\).
Thus the experiment demonstrates a low-rank basis-compression effect within
ordinary Rayleigh--Ritz approximation.  The gain is largest when the active
labels select substantially different reference-sector geometries.

\subsection{Positive-measure harvesting of projected matrices}
\label{subsec:numerics-e2-e2b}

We next test the correlation-matrix construction of
\S\ref{subsec:correlation-primary} using the same Hamiltonian form, now
with
\begin{equation}
        \Omega=1,
        \qquad
        \epsilon=0.25,
        \qquad
        t=0.35,
        \qquad
        g=0.5,
        \qquad
        a=0.1.
        \label{eq:e2-parameters}
\end{equation}
The GEVP is solved at \(n_0=2\) and \(n=8\), using
Eqs.~\eqref{eq:correlation-matrix} and \eqref{eq:GEVP-main}.

\subsubsection{Deterministic projected reference}

We first construct \(C(n)\) deterministically from a high-cutoff oscillator
representation whose results are stable under further cutoff enlargement.
This separates projected-space error from Monte Carlo error.  At \(n_0=2\)
and \(n=8\), the errors in the two lowest GEVP energies are reported in
Table~\ref{tab:e2-deterministic-errors}.  Here
\(\dim\Vcal_R\) is the total dimension of the correlation-matrix trial space.

\begin{table}[ht]
\centering
\begin{tabular}{c|cc}
\hline
\(\dim\Vcal_R\) &
\(\left|E_0^{\rm GEVP}-E_0^{\rm exact}\right|\) &
\(\left|E_1^{\rm GEVP}-E_1^{\rm exact}\right|\) \\
\hline
\(8\)  & \(4.86\times10^{-5}\) & \(8.05\times10^{-5}\) \\
\(16\) & \(5.87\times10^{-8}\) & \(8.60\times10^{-8}\) \\
\hline
\end{tabular}
\caption{Deterministic GEVP energy errors at \(n_0=2\) and \(n=8\).}
\label{tab:e2-deterministic-errors}
\end{table}

The errors decrease by nearly three orders of magnitude when the trial-space
dimension is increased from \(8\) to \(16\).  The same calculation also
illustrates the Euclidean-time conditioning trade: increasing \(n_0\)
suppresses higher-energy directions but makes the remaining metric more
ill-conditioned.  At trial-space dimension \(16\),
\(\operatorname{cond}C(n_0)\) increases from approximately \(4.76\) at
\(n_0=1\) to \(514.25\) at \(n_0=4\).  
The stochastic calculation below therefore monitors the metric spectrum and
condition number.

\subsubsection{Stochastic CEH correlation matrices}

To produce the same type of matrix data from a positive reference measure, we
take the harmonic oscillator as the reference sector.  Write its normalized
eigenfunctions as
\[
        \varphi_r(x)=\varphi_0(x)p_r(x),
\]
where \(p_r\) is the normalized Hermite-polynomial factor relative to the
ground-state probability density \(|\varphi_0(x)|^2\,dx\).

A stationary Ornstein--Uhlenbeck process with this invariant density generates
the oscillator paths.
Equivalently, with \(X_0\sim|\varphi_0|^2dx\) and independent
\(\xi_j\sim\mathcal N(0,1)\),
\[
        X_{j+1}
        =
        \e^{-\Omega a}X_j
        +
        \sqrt{\frac{1-\e^{-2\Omega a}}{2\Omega}}\,\xi_j.
\]
For a path
\(X=(X_0,\ldots,X_n)\),
the active-sector Hamiltonian is
\begin{equation}
        A(X_j)
        =
        \epsilon\sigma_z+t\sigma_x+gX_j\sigma_z,
        \label{eq:e2-active-matrix}
\end{equation}
and its conditioned Euclidean propagation is approximated by
\begin{equation}
        U_{\rm active}[X]
        =
        \e^{-aA(X_n)/2}
        \e^{-aA(X_{n-1})}\cdots\e^{-aA(X_1)}
        \e^{-aA(X_0)/2} .
        \label{eq:e2b-symmetric-active-product}
\end{equation}
The correlation-matrix
estimator is
\begin{equation}
        C_{r\alpha,s\beta}(n)
        =
        \e^{-\Omega na/2}
        \mathbb E_{\rm ref}
        \left[
        p_r(X_n)p_s(X_0)
        \bigl(U_{\rm active}[X]\bigr)_{\alpha\beta}
        \right].
        \label{eq:e2b-stochastic-estimator}
\end{equation}
Every sampled reference path therefore contributes a finite active-sector
matrix.  The active trace is not inserted into the probability measure.

The stochastic calculation uses \(R_b=2\) oscillator states per active
channel, giving a four-dimensional trial space.
At \(30000\) paths per seed and ten
independent seeds, the extracted energies are shown in
Table~\ref{tab:e2b-summary}.  The stochastic values agree with the
deterministic projected reference within approximately one standard error.
The stochastic--deterministic projected difference contains sampling and
symmetric-splitting error, while the further difference from exact
diagonalization also contains finite-trial-space error.

\begin{table}[ht]
\centering
\begin{tabular}{c|ccc}
\hline
level &
stochastic GEVP &
deterministic projected &
exact diagonalization \\
\hline
\(E_0\) &
\(-0.017842\pm0.000885\) &
\(-0.016893\) &
\(-0.020021\) \\
\(E_1\) &
\(0.662122\pm0.002463\) &
\(0.664953\) &
\(0.643571\) \\
\hline
\end{tabular}
\caption{Stochastic, projected, and exact GEVP energies in the \(R_b=2\) trial space.}
\label{tab:e2b-summary}
\end{table}

The stochastic mean condition number of \(C(n_0)\) is approximately \(1.49\).
No metric eigenvalue cut was required, no nonpositive retained transfer
eigenvalue was encountered, and all GEVP solves succeeded.

\begin{figure}[ht]
        \centering
        \includegraphics[width=0.75\textwidth]
        {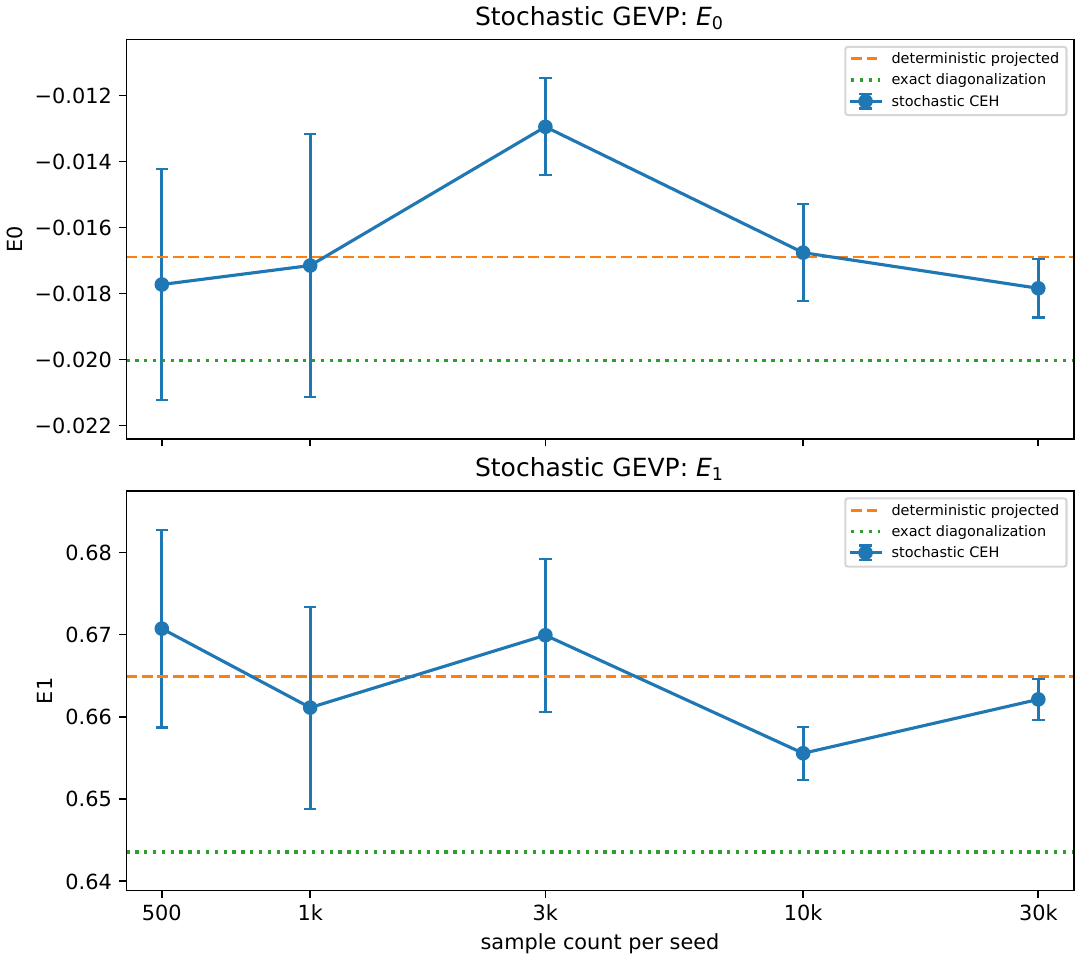}
	\caption{Stochastic GEVP energies versus sample count.}
	\label{fig:e2b-stochastic-gevp}
\end{figure}

This closes the minimal stochastic CEH loop in the oscillator model:
\[
        \text{paths from a positive reference measure}
        \longrightarrow
        \widehat C_{IJ}(n)
        \longrightarrow
        \text{projected spectral data}.
\]
It does not imply that the same estimator remains well conditioned at large
active rank or spatial volume.

\subsubsection{Comparison with a PDMS-style estimator}

We also compare the CEH estimator with a PDMS-style construction on the same
two-channel Hamiltonian.  The PDMS-style branch estimates projected thermal
density and energy matrices,
\begin{equation}
        Z_{ac}(\beta)
        =
        \langle\psi_a,\e^{-\beta H}\psi_c\rangle,
        \qquad
        E_{ac}(\beta)
        =
        \langle\psi_a,H\e^{-\beta H}\psi_c\rangle,
        \label{eq:e2c-pdms-matrices}
\end{equation}
and extracts projected energies from the corresponding generalized matrix
problem.
The CEH branch instead estimates reference-sector correlation matrices
with active-sector matrix observables and applies the GEVP
\eqref{eq:GEVP-main}.
The two branches therefore share projected linear algebra
but do not sample the same stochastic object.

Table~\ref{tab:e2c-pdms-vs-ceh} reports the largest common sampling budget,
\(8000\) samples per seed and five independent seeds.  In this test the
PDMS-style estimator gives smaller projected energy errors, whereas the CEH
metric is better conditioned and its selected diagonal entry has a smaller
seed-to-seed standard error.
The reported condition numbers are those of \(Z\) and \(C(n_0)\) for the
PDMS-style and CEH branches, respectively; their uncertainties are
seed-to-seed standard errors.

\begin{table}[ht]
\centering
\begin{tabular}{l|cc}
\hline
quantity &
PDMS-style &
CEH correlation estimator \\
\hline
\(E_0\) error
& \(1.47\times10^{-4}\)
& \(4.77\times10^{-3}\) \\
\(E_1\) error
& \(1.75\times10^{-3}\)
& \(1.74\times10^{-2}\) \\
metric condition number
& \(39.1\pm3.35\)
& \(1.483\pm0.011\) \\
selected diagonal-entry SE
& \(6.61\times10^{-4}\)
& \(2.49\times10^{-4}\) \\
regularization or failed solve
& none
& none \\
\hline
\end{tabular}
\caption{PDMS-style and CEH estimates at \(8000\) samples per seed.}
\label{tab:e2c-pdms-vs-ceh}
\end{table}

The uncertainty row is the seed-to-seed standard error of a preselected
diagonal entry: \(Z_{00}\) for the PDMS-style projected density matrix and
\(C_{00}(n_0)\), with \(n_0=2\), for the CEH correlation metric.  The flattened
index \(0\) corresponds to \((r,\alpha)=(0,0)\) under the convention
\(I=2r+\alpha\).  Since the two methods form different projected stochastic
objects, this is a variance diagnostic rather than an element-by-element
comparison of identical observables.  More generally, projected-energy
accuracy, metric conditioning, and matrix-entry variance need not favor the
same estimator.

\subsection{Finite-density continuation in a Hubbard ring}
\label{subsec:numerics-hubbard}

The final experiment uses a periodic four-site spinful Hubbard ring and is
organized into two parallel diagnostics.
The determinant branch treats the full finite-density model in the standard
scalar-weight formulation and provides a finite-budget sign-stress diagnostic.
The projected branch tests the
operator-split \(H(\mu)=H(0)-\mu N\) of
\S\ref{subsec:finite-density-after-projection}.
A trial space is selected using only chemical potentials below the target,
\(H(0)\) and \(N\) are projected onto that space, and the target chemical
potential is then applied through \(H_R-\mu N_R\).  The test therefore isolates
whether the non-target space contains the number sectors required for
finite-density observables.

We use the shifted repulsive Hubbard Hamiltonian
\begin{equation}
        H(\mu)
        =
        -t_{\rm hop}
        \sum_{\langle i,j\rangle,\sigma}
        \left(
        c_{i\sigma}^\dagger c_{j\sigma}
        +
        c_{j\sigma}^\dagger c_{i\sigma}
        \right)
        +
        U\sum_i
        \left(n_{i\uparrow}-\frac12\right)
        \left(n_{i\downarrow}-\frac12\right)
        -
        \mu N,
        \label{eq:hubbard-shifted}
\end{equation}
with
\[
        N=\sum_{i,\sigma}n_{i\sigma}.
\]
Unless otherwise stated, the parameters are
\begin{equation}
        L=4,
        \qquad
        t_{\rm hop}=1,
        \qquad
        U=8,
        \qquad
        \beta=8,
        \qquad
        \mu=2,
        \qquad
        \Delta\tau=0.25.
        \label{eq:hubbard-parameters}
\end{equation}
The full Fock-space dimension is \(4^L=256\).

\subsubsection{Determinant sign benchmark}

We use the standard discrete Hubbard--Stratonovich formulation
\cite{AssaadEvertz2008}.  The determinant branch samples the absolute weight
\begin{equation}
        \left|
        \det M_\uparrow[s]\det M_\downarrow[s]
        \right|,
        \label{eq:hubbard-absolute-weight}
\end{equation}
where \(s\) denotes the auxiliary-field history.  The determinant sign
\begin{equation}
        \operatorname{sgn}[s]
        =
        \operatorname{sgn}
        \left(
        \det M_\uparrow[s]\det M_\downarrow[s]
        \right)
        \label{eq:hubbard-sign}
\end{equation}
is restored in observables.  The particle--hole-symmetric \(\mu=0\)
calculation has unit sign and serves as a control.

The target calculation uses ten independent chains.  For each chain, the
first \(25\) sweeps are discarded, followed by \(90\) post-warm-up sweeps.
Here \(N_\tau=\beta/\Delta\tau=32\), and one sweep attempts
\(L N_\tau=128\) local Hubbard--Stratonovich-field flips.  The sign is recorded
every third sweep, giving \(30\) retained measurements per chain and \(300\)
in total.

For an independent chain \(r\), let
\[
        \bar s_r
        =
        \frac{1}{N_r}\sum_{k=1}^{N_r}s_{r,k}.
\]
We report the pooled signed mean, the mean value of \(|\bar s_r|\), and
\begin{equation}
        q_{\rm sign}
        =
        \frac{1}{N_{\rm seed}}
        \sum_{r=1}^{N_{\rm seed}}
        |\bar s_r|^2,
        \label{eq:e5-sign-diagnostic}
\end{equation}
which we use as a descriptive mean-squared chain-sign diagnostic.  Because all
chains contain the same number of retained measurements, the average of the
chain means is also the pooled mean over all retained signs.

At the target point \eqref{eq:hubbard-parameters}, we have
\begin{equation}
        \overline{\operatorname{sgn}}
        =
        0.020\pm0.099,
        \qquad
        \frac{1}{N_{\rm seed}}\sum_r|\bar s_r|
        =
        0.220\pm0.067,
        \qquad
        q_{\rm sign}=0.0893\pm0.0511.
        \label{eq:hubbard-sign-results}
\end{equation}
The uncertainties are standard errors across the ten independent chain
statistics.  The pooled signed mean is consistent with zero, and the
chain-level statistics show substantial sign cancellation.  Since each chain
contains only \(30\) retained measurements and no autocorrelation analysis is
used here, \(q_{\rm sign}\) should not be interpreted as an unbiased
asymptotic effective-sample-size fraction.  This calculation is a finite-budget
determinant-sign stress diagnostic.

\subsubsection{Non-target active spaces}

Write
\begin{equation}
        H_0=H(0),
        \qquad
        H(\mu)=H_0-\mu N.
        \label{eq:hubbard-H0}
\end{equation}
Since \([H_0,N]=0\), the eigenvectors may be labeled simultaneously by energy
and particle number:
\begin{equation}
        H_0|E_j,N_j\rangle
        =
        E_j|E_j,N_j\rangle,
        \qquad
        N|E_j,N_j\rangle
        =
        N_j|E_j,N_j\rangle.
        \label{eq:hubbard-common-eigenbasis}
\end{equation}
Their finite-density energies are
\begin{equation}
        E_j(\mu)=E_j-\mu N_j.
        \label{eq:hubbard-tilted-levels}
\end{equation}
The issue in this benchmark is therefore number-sector coverage: a low-energy
basis selected at \(\mu=0\) need not contain the sectors whose tilted levels
become thermally relevant at the target \(\mu=2\).

The simplest non-target space is
\begin{equation}
        \Vcal_R^{(0)}
        =
        \operatorname{span}
        \left\{
        |E_1,N_1\rangle,\ldots,|E_R,N_R\rangle
        \right\},
        \label{eq:hubbard-mu0-space}
\end{equation}
where the states are ordered by their energies at \(\mu=0\).  On this space we
form
\begin{equation}
        H_R=U_R^\dagger H_0U_R,
        \qquad
        N_R=U_R^\dagger NU_R,
        \qquad
        H_R(\mu)=H_R-\mu N_R.
        \label{eq:hubbard-projected-HN}
\end{equation}

We also construct a non-target multi-\(\mu\) space using
\begin{equation}
        \Mcal_{\rm train}=\{0,0.5,1\},
        \label{eq:hubbard-training-mus}
\end{equation}
which excludes the target value.  For each
\(\mu'\in\Mcal_{\rm train}\), let \(\psi_j^{(\mu')}\) denote the eigenvectors
of \(H_0-\mu' N\), ordered by nondecreasing eigenvalue.  The candidates are
interleaved by eigenlevel across the three training values and truncated after
a total of \(R\) vectors.  Thus \(R\) is a requested candidate budget, not a
number of states selected separately at each \(\mu'\).

Because candidates selected at different training values may coincide or be
linearly dependent, we reduce their span by a dense SVD of the
\(D\times R\) candidate matrix \(B_R\), with \(D=256\).  We retain the
singular directions satisfying
\[
        \sigma_a
        >
        10^{-10}\max\{D,R\}\,\sigma_{\max}
\]
and use the corresponding left singular vectors as an orthonormal active
basis.  Denoting its dimension by \(r_R\leq R\), the requested budgets
\(R=2,4,8,16,32,64\) give
\[
        r_R=1,2,4,6,11,26.
\]
The reduction from \(64\) candidates to \(26\) independent directions reflects
substantial redundancy across the training values.

Because \([H_0,N]=0\), varying \(\mu'\) only reorders the common eigenspaces
through the tilted energies \(E_j-\mu'N_j\).  Separate diagonalizations may
return different bases within degenerate subspaces, so linear dependence is
tested at the level of the candidate span rather than by matching eigenvector
labels.  The construction is therefore a sector-reordering and rank-filtering
test, not continuation along a changing eigenvector manifold.

For either construction, let \(\Vcal_R\) denote the retained space associated
with requested candidate budget \(R\).  The projected thermal density per site
is
\begin{equation}
        \rho_R(\beta,\mu)
        =
        \frac{1}{L}
        \frac{
        \Tr_{\Vcal_R}
        \left[
        N_R\e^{-\beta(H_R-\mu N_R)}
        \right]
        }{
        \Tr_{\Vcal_R}
        \left[
        \e^{-\beta(H_R-\mu N_R)}
        \right]
        }.
        \label{eq:hubbard-projected-density}
\end{equation}
At the target point, exact diagonalization gives
\(\rho_{\rm exact}=1.0003059\).  The normalized thermal weights are
\(0.99877638\) in \(N=4\), \(0.001223579\) in \(N=5\), and
\(3.752\times10^{-8}\) in \(N=6\), with all other sectors negligible at the
displayed precision.  Thus a projected space that contains the target
\(N=4\) ground state but no neighboring sectors can reproduce the ground-state
energy while predicting \(\rho_R=1\), leaving a density error of approximately
\(3.06\times10^{-4}\).

Figure~\ref{fig:e5-hubbard-compact} shows the dependence on the requested
candidate budget \(R\).  The \(\mu=0\)-only space remains confined to \(N=4\)
through \(R=16\).  At \(R=32\), it first includes \(N=3,4,5\), and the density
error drops from approximately \(3.059\times10^{-4}\) to
\(9.195\times10^{-8}\).  The ground-state energy converges earlier because the
target ground state is represented before the weakly occupied neighboring
sectors are complete.

To monitor number-sector leakage, define
\begin{equation}
        \eta_N(R)
        =
        \frac{
        \|[H_R,N_R]\|_F
        }{
        \|H_R\|_F\|N_R\|_F
        }.
        \label{eq:hubbard-number-commutator}
\end{equation}
The retained spaces in this deterministic construction are numerically
invariant under \(N\), so \(\eta_N(R)\) vanishes up to roundoff and SVD
rank-filtering errors.  It is therefore a consistency check rather than a
nontrivial conservation result.  In a stochastic correlation-matrix
implementation, the same diagnostic would detect projection-induced
number-sector leakage.

\begin{figure}[ht]
        \centering
        \includegraphics[width=0.75\textwidth]
        {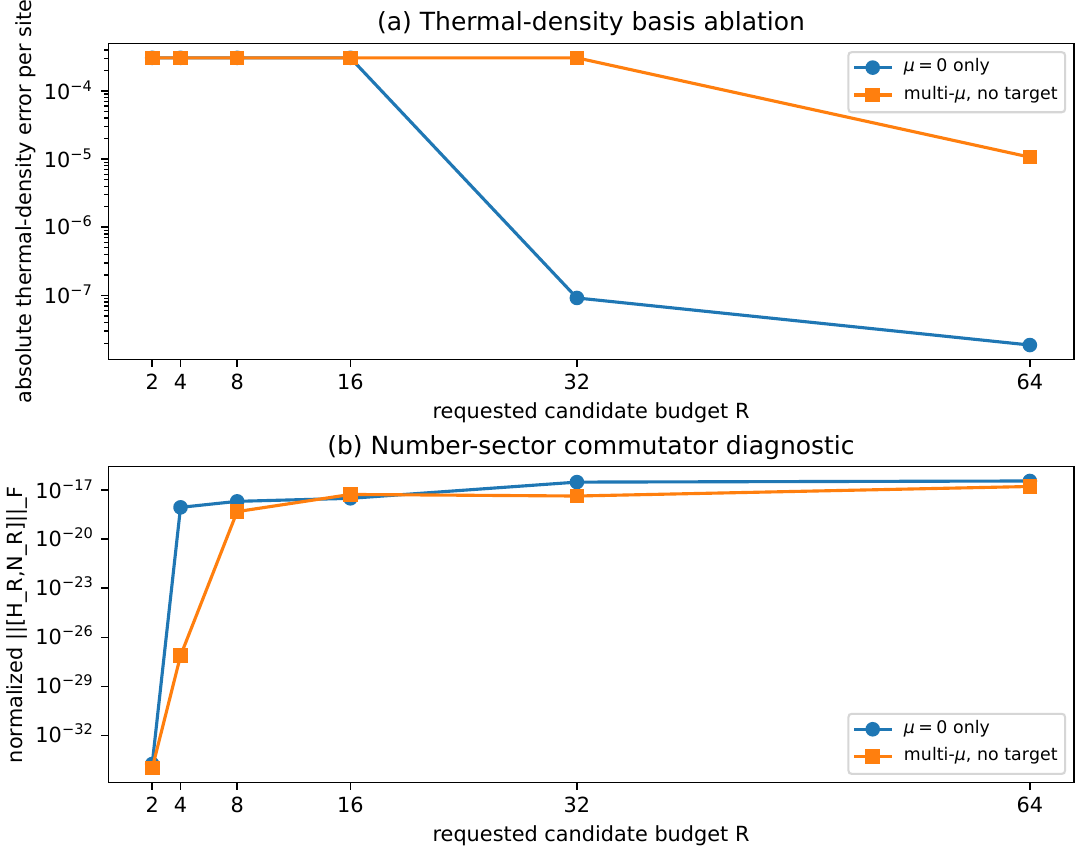}
        \caption{Hubbard active-space diagnostics.}
        \label{fig:e5-hubbard-compact}
\end{figure}

Table~\ref{tab:e5c-hubbard-ablation} compares the two constructions at the
same requested candidate budget \(R=64\), not at the same retained dimension.
The \(\mu=0\)-only space retains all \(64\) directions, whereas the
multi-\(\mu\) space retains \(r_{64}=26\), comprising \(4\), \(16\), and \(6\)
directions in the \(N=3,4,5\) sectors, respectively.  Both spaces recover the
target ground-state energy, but the larger \(\mu=0\)-only space gives the more
accurate thermal density.  The multi-\(\mu\) construction is included as a
test of candidate selection and redundancy, not as an accuracy improvement at
fixed requested budget.

\begin{table}[ht]
\centering
\small
\begin{tabular}{l|ccc}
\hline
basis &
retained rank &
\(\left|\rho_R-\rho_{\rm exact}\right|\) &
\(\left|E_{0,R}-E_0\right|\) \\
\hline
\(\mu=0\) only
&
\(64\)
&
\(1.875\times10^{-8}\)
&
\(<10^{-14}\)
\\
multi-\(\mu\), no target
&
\(26\)
&
\(1.069\times10^{-5}\)
&
\(<10^{-14}\)
\\
\hline
\end{tabular}
\caption{Hubbard active spaces at requested budget \(R=64\).}
\label{tab:e5c-hubbard-ablation}
\end{table}

The two Hubbard branches have deliberately limited roles.  The determinant
branch provides a finite-budget sign-stress diagnostic, while the active-space
branch tests number-sector coverage and rank filtering in a controlled
deterministic setting.  The latter is not scalable, since its candidates are
obtained by exact diagonalization of the small reference Hamiltonians, and it
does not address the harder task of learning active directions from Euclidean
correlation data.


\section{Concluding remarks}
\label{sec:conclusion}

CEH places the reference/active split before the residual sector is collapsed
into a scalar Euclidean weight.  A tractable reference calculation supplies
projected correlation or transfer matrices, while the residual degrees of
freedom remain operator-valued and are reduced to a finite Hamiltonian model;
any deferred finite-density dependence is then applied through the reduced
operators.  At finite rank this is an approximation: direct Galerkin
compression commutes with adding \(-\mu N\) on a fixed trial space, whereas a
transfer-logarithm Hamiltonian is only an effective generator of the resolved
Euclidean dynamics.  Reliability therefore requires rank and Euclidean-time
enlargement, metric conditioning, stochastic matrix uncertainty, and
number-sector diagnostics.

The oscillator tests separated basis and stochastic errors.  A channel-adapted
basis gave markedly better low-rank compression than an undressed product
basis, although the latter caught up at larger cutoff.  The positive-reference
calculation reproduced the deterministic projected GEVP energies within
statistical errors.  In the controlled PDMS-style comparison, PDMS gave smaller
energy errors, while the CEH metric was better conditioned and its selected
diagonal entry had a smaller seed-to-seed standard error.

The Hubbard ring combined a finite-budget determinant-sign stress test with a
separate sector-coverage calculation.  The ground-state energy converged once
the dominant canonical sector was present, but the thermal density required
weakly occupied neighboring sectors.  Because exact diagonalization supplied
the trial spaces, this benchmark is neither scalable nor an end-to-end
stochastic CEH solution of the Hubbard sign problem.

The benchmarks support the proposed trade from scalar sign reweighting to a
projected matrix problem with visible errors, provided the active dynamics are
low-rank approximable and the required matrix data are harvestable with usable
overlap, variance, and conditioning.  The next decisive test is a spatially
extended model in which reference-sector Euclidean data construct an active
basis that is enlarged while rank, metric resolution, stochastic signal, and
number-sector coverage are tracked with system size and density.

\section*{Data availability}

The numerical data underlying the tables and figures, together with the
essential analysis scripts, are available from the author upon reasonable
request.

\section*{Acknowledgements}

I thank Professors Ming Mei and Zejia Wang for hosting my visit to Jiangxi
Normal University, and the physics faculty members at the National University
of Mongolia who participated in the 2024--2025 Theoretical and Mathematical
Physics seminar.  This work was partially supported by an NSERC Discovery
Grant.

\end{document}